\documentclass[11pt,letterpaper]{article}
\usepackage{cite,comment}
\usepackage{tikz}\usetikzlibrary{matrix,decorations.pathreplacing,positioning}
\usepackage[colorlinks,citecolor=blue,linkcolor=blue,urlcolor=blue]{hyperref}
\usepackage{amsmath,amsthm,amssymb,cases}
\usepackage{enumitem}

\usepackage{calrsfs}

\usepackage{empheq}

\usepackage{authblk}
\author[1]{Gianira N. Alfarano\thanks{G.~N.~Alfarano is supported by Swiss Science Foundation grant n.~188430.}}
\affil[1]{University of Zurich, Switzerland}

\author[2]{Alberto Ravagnani\thanks{A.~Ravagnani is in part supported by Dutch Research Council grant VI.Vidi.203.045.}}
\affil[2]{Eindhoven University of Technology, the Netherlands}

\author[3]{Emina Soljanin\thanks{E.~Soljanin is in part supported by NSF Award CIF-2122400.}}
\affil[3]{Rutgers University, NJ, USA}

\usepackage[margin=3cm]{geometry}

\definecolor{myg}{RGB}{220,220,220}
 
\theoremstyle{definition}
\newtheorem{theorem}{Theorem}[section]
\newtheorem{corollary}[theorem]{Corollary}
\newtheorem{proposition}[theorem]{Proposition}
\newtheorem{lemma}[theorem]{Lemma}
\newtheorem{definition}[theorem]{Definition}
\newtheorem{example}[theorem]{Example}
\newtheorem{notation}[theorem]{Notation}
\newtheorem{remark}[theorem]{Remark}

\newcommand*{\myproofname}{Proof of the claim}

\usepackage{hyperref}

\newcommand{\numberset}{\mathbb}

\newcommand{\R}{\numberset{R}}

\newcommand{\F}{\numberset{F}}

\newcommand{\mC}{\mathcal{C}}

\newcommand{\mR}{\mathcal{R}}

\newcommand{\wt}{\textnormal{wt}}

\newlength{\mynodespace}
\newcommand{\all}{\textnormal{all}}

\begin{document}

\title{\textbf{Dual-Code Bounds on Multiple Concurrent \\ (Local) Data Recovery}}

\date{}






\maketitle

\begin{abstract}
We are concerned with linear redundancy storage schemes regarding their ability to provide concurrent (local) recovery of multiple data objects. This paper initiates a study of such systems within the classical coding theory. We show how we can use the structural properties of the generator matrix defining the scheme to obtain a bounding polytope for 
the set of data access rates the system can support. 
We derive two dual distance outer bounds, which are sharp for some large classes of matrix families.
\end{abstract}

%

\bigskip

\section{Introduction}
Distributed computing systems rely on their storage layers to provide data access services for executing applications. Thus, the system's overall performance depends on the underlying storage system's data access performance.
Distributed storage systems strive to maximize the number of concurrent data access requests they can support with fixed resources. Replicating data objects according to their relative popularity and access volume helps achieve this goal. However, these quantities are often unpredictable. In emerging applications such as edge computing, the expected number of users and their data interests fluctuate, and data storage schemes should support such dynamics \cite{Edge:YadgarKAS19}. Erasure-coding has emerged as an efficient and robust form of redundant storage, which can 
 flexibly handle skews in the request rates.
 
Recent work on redundant distributed storage access 
introduced the notion of the {\it service rate region} of a redundancy scheme that includes all data access requests that the system can serve
\cite{service:aktas2021jkks,service:KazemiKSS21g,service:KazemiKSS20g,service:KazemiKSS20c,service:AndersonJJ18,service:AktasAJ17}.
To understand this concept, consider a distributed system that stores $k$ different data objects by encoding them into $n$ and storing the $n$ coded objects on $n$ different nodes. Each of the $n$ nodes can serve requests at a rate $\mu$ (i.e., has service capacity $\mu$). The system can serve requests to access the $k$ data objects that arrive at rates $\lambda_1$, $\lambda_2$, \dots , $\lambda_k$ if each request can be routed to a group of nodes that can jointly fulfill the request, and the total request rate allocated to each node does not exceed its service capacity $\mu$. We call the set of such request vectors ($\lambda_1$, $\lambda_2$, \dots , $\lambda_k$) the \emph{service rate region} of a coded distributed system.

 The overview paper \cite{service:aktas2021jkks} postulated the service rate region as an important consideration in the design of erasure-coded distributed systems. It highlights several open problems that can be grouped into two broad threads: 1) characterizing the service rate region of a given code and finding the optimal request allocation, and 2) designing the underlying erasure code for a given service rate region. The paper argued that the presented problems not only require expertise from different areas, but have also already been addressed in those areas in some special forms and under different names. Moreover, it explained how some problems associated with the service rate region generalize previously studied distributed problems such as batch codes, codes with locality and availability, and private information retrieval \cite{BatchCodesAndTheirApps:IshaiKO04,Batch:RietST18,  batchPIR:Skachek18, PIR:Fazeli15-1,PIR:Fazeli15-2,Ernvall:14,Ernvall:14a,Gopalan:14,OptimalLRC}. 
 
 The numerous open problems described in \cite{service:aktas2021jkks} (in both groups mentioned above) could be seen as performance analysis and networking problems as well as coding theory and data allocation problems. These problems could be addressed by a wide variety of scientists according to their interests and expertise. 
 \\[0.5ex]
 \textbf{Our contribution.}
 The goal of this paper is to initiate a study of redundancy schemes within the classical coding theory. 
 We focus on characterizing the service rate region of a storage scheme defined by a rank~$k$ generator matrix $G$, which is
 a convex polytope in~$\R^k$.
 Our ultimate goal is
 to establish a series of inequalities that, when combined, cut out the service region.
 This paper makes a first step towards this goal by showing 
 how some structural properties of~$G$ 
 can be used to find a 
 polytope that contains the service rate region, giving an \textit{outer bound} for the latter. More precisely, 
 we establish a Total Capacity Bound for the service rate region determined by a matrix~$G$, as well as two Dual Distance Bounds that take into account different structural properties of~$G$.
 We also show that the {\it bounding} polytope we find coincides with the service rate region in some special cases, e.g.
 for some large classes of MDS codes. 

This paper is organized as follows. Sec.~\ref{sec:ii} defines the problem. Sec.~\ref{sec:iii} shows some properties of the recovery sets of a linear redundancy scheme and introduces a way to compare storage schemes.
Sec.~\ref{sec:iv} derives the Total Capacity Bound for the service rate region and two Dual Distance Bounds. Sec.~\ref{sec:v} outlines 
future work plans.

\section{Distributed Coded Recovery Systems}
\label{sec:ii}

In this section we establish the notation for 
the rest of the paper, define distributed coded systems and their service rate region. 

\begin{notation}
Throughout the paper, $\F_q$ denotes the finite field with $q$ elements where $q$ is a prime power.
We work with integers $n > k \ge 2$, a real number $\mu \ge 1$,
and a fixed matrix
$G \in \F_q^{k \times n}$ of rank $k$. We assume that $G$ has no all-zero column and
denote by~$G^j$ 
its $j$th column.
\end{notation}

We consider coded distributed systems where~$k$ data objects are linearly encoded into $n$ objects stored on $n$ servers. Each server stores exactly one object and the objects are elements of $\F_q$. 

Such a distributed coded system is fully specified by a rank $k$ matrix $G \in \F_q^{k \times n}$, which we call the \textbf{generator matrix} of the system. If $(x_1,\dots,x_k) \in \F_q^k$ is the $k$-tuple of objects to be stored,
then the $j$th server stores the $j$th component of the vector
$$(x_1,\dots,x_k) \cdot G \in \F_q^n.$$
We say
that the matrix $G$ is \textbf{systematic} if its first $k$ columns form the identity $k \times k$ matrix. 

We consider $k$-tuples $(\lambda_1,\ldots,\lambda_k) \in \R^k$ of rate requests to access the $k$ objects. More precisely, $\lambda_i$ is the rate request for the $i$th object.
Each of the $n$ servers can serve request at the rate of at most $\mu \ge 1$. The parameter $\mu$ is called the server's \textbf{capacity}. 

Each user gets assigned to a set of servers that, together, allow recovering the desired object. An object can be recovered from different server sets, which motivates the following terminology.

\begin{definition} \label{maindef}
For $i \in \{1,\dots,k\}$, let
$$\mR^\all_i:=\{R \subseteq \{1,\dots,n\} \mid 
e_i \in \langle G^j \mid j \in R\rangle\},$$
where $e_i \in \F_q^k$ denotes the $i$th
standard basis vector and $\langle G^j \mid j \in R\rangle$ is the span of the columns of~$G$ indexed by $R$. The elements of $\mR^\all_i$ are the \textbf{recovery sets} for the $i$th object.
\end{definition}

Note that, in the above definition, we have 
$\mR^\all_i \neq \emptyset$ for all $i \in \{1,\ldots,k\}$. This is a simple consequence of the fact that $G$ has rank $k$. Moreover, $R \neq \emptyset$ for all $i\in\{1,\dots,k\}$ and $R \in \mR^\all_i$. We use the superscript
``all'' to indicate that $\mR_i^\all$ contains \textit{all} the recovery sets for the symbol $i$.

\begin{example}\label{ex:MDScode}
Let $k=2$, $n=4$, $q=3$ and 
$$G := \begin{pmatrix}
1 & 0 & 1 & 1 \\
0 & 1 & 2 & 1
\end{pmatrix}\in\F_3^{2\times 4}.$$
 Then, we have 
 \begin{align*}
    \mR^\all_1&=\{\{1\},\{2,3\}, \{2,4\}, \{3,4\}, \{2,3,4\}\},\\
     \mR^\all_2&=\{\{2\},\{1,3\}, \{1,4\}, \{3,4\}, \{1,3,4\}\}.
\end{align*}
\end{example}

When designing a recovery system starting from~$G$, not all recovery sets need to be considered.

\begin{definition}
A \textbf{recovery} $G$-\textbf{system} is a $k$-tuple
$\mR=(\mR_1,\ldots,\mR_k)$ of subsets of $\{1,\ldots,n\}$ with $\mR_i \subseteq \mR_i^{\textnormal{all}}$ and $\mR_i \neq \emptyset$ for all $i \in \{1,\ldots,k\}$.
\end{definition}

The service rate region of a recovery $G$-system $\mR$ is the set of all request rate tuples $(\lambda_1,\ldots,\lambda_k) \in \R^k$ that can be served by the system.

\begin{definition} \label{def:system}
Let $\mR=(\mR_1,\ldots,\mR_k)$ be
a recovery $G$-system.
The \textbf{service rate region} associated with~$\mR$ and~$\mu$ is the set of
all $(\lambda_1,\ldots,\lambda_k) \in \R^k$ for which there exists a collection of real numbers $$\{\lambda_{i,R} \mid i \in \{1,\ldots,k\}, \, R \in \mR_i\}$$
with the following properties:
\begin{empheq}[left = \empheqlbrace]{alignat=3}
        \sum_{R \in \mR_i}\lambda_{i,R} &= \lambda_i &&\textnormal{ for } 1\leq i\leq k, \label{c1}\\
        \sum_{i=1}^k \sum_{\substack{R \in \mR_i \\ j \in R}} \lambda_{i,R} &\leq \mu  \ &&\textnormal{ for } 1\leq j\leq n, \label{c2} \\
        \lambda_{i,R} &\geq 0  \ &&\textnormal{ for } 1\leq i\leq k, \, R \in \mR_i. \label{c3}
    \end{empheq}
A collection $\{\lambda_{i,R}\}$
of real numbers that satisfy the above three properties is called a  \textbf{feasible allocation} for $(\mR,\mu)$. The service rate region associated with $\mR$ and $\mu$ is denoted by $$\Lambda(\mR,\mu) \subseteq \R^k.$$
\end{definition}

\begin{remark}
We have 
$\Lambda(\mR,\mu) = \mu \cdot \Lambda(\mR,1)$ 
for any recovery $G$-system~$\mR$. Therefore, we will simply call $\Lambda(\mR,1)$ the \textbf{service rate region} of $G$.
\end{remark}

 Our main goal is to establish outer bounds for the service rate region $\Lambda(\mR,1)$ of a recovery $G$-system~$\mR$, in the form of a \textit{bounding} polytope.

\begin{example} \label{ex:cont}
Let $G$ be as in Example~\ref{ex:MDScode}.
Corollary~\ref{cor:MDSregion} will give us the bounding polytope 
for $\Lambda(\mR^\textnormal{all},1)$ depicted in Figure~\ref{fig:4_2_rate_region}.
The outer bound is sharp.
\begin{figure}[hbt]
\begin{center}
\includegraphics[scale=0.7]{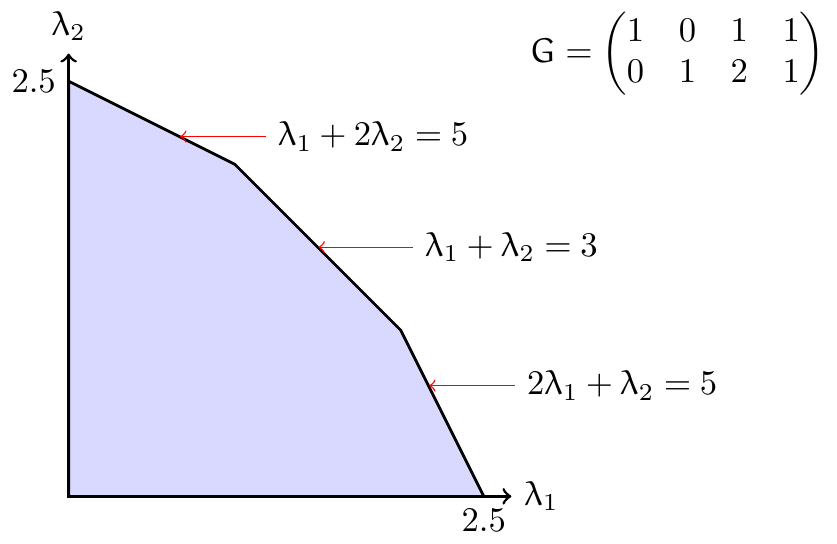}
\end{center}
\caption{Service rate region for the $G$-system in Example~\ref{ex:cont}. The lines bounding the polytope are obtained from Corollary~\ref{cor:MDSregion}. \label{fig:4_2_rate_region}}
\end{figure}
\end{example}

\section{Minimal Recovery Sets}
\label{sec:iii}

A natural question is how the service rate region changes when changing the recovery $G$-system.
An immediate observation is the following. The result easily follows from Definition~\ref{def:system} and is therefore left to the reader.

\begin{proposition}\label{prop:subsystem}
Suppose that $\mR=(\mR_1,\ldots,\mR_k)$ and 
$\mR'=(\mR'_1,\ldots,\mR'_k)$ are recovery $G$-systems with $\mR'_i \subseteq \mR_i$ for all 
$i \in \{1,\ldots,k\}$. Then
$\Lambda(\mR,1) \supseteq \Lambda(\mR',1)$. In particular,
$\Lambda(\mR,\mu) \subseteq 
    \Lambda(\mR^\all,\mu)$
for any recovery $G$-system $\mR$.
\end{proposition}

The service rate region, however, does not change when selecting from $\mR^{\textnormal{all}}$ the recovery sets that are 
minimal with respect to inclusion, in the following precise sense.

\begin{definition}
A set $R \in \mR^\all_i$ is 
\textbf{$i$-minimal} if there is no $R' \in \mR^\all_i$ with $R' \subsetneq R$. We define~$\mR^{\textnormal{min}}$ as the recovery $G$-system defined, for all $i$, by
$$\mR_i^{\textnormal{min}}:=\{R \in \mR^\all_i \mid R \mbox{ is $i$-minimal}\}.$$
\end{definition}

\begin{proposition}\label{lostesso}
We have
$\Lambda(\mR^{\textnormal{min}},1) = \Lambda(\mR^\all,1)$.
\end{proposition}

Therefore, when studying the service rate region of the system defined by $G$ one may or may not restrict to the minimal recovery sets.

\begin{proof}[Proof of Proposition~\ref{lostesso}]
The inclusion ``$\subseteq$'' 
follows directly from Proposition~\ref{prop:subsystem}.
To prove the inclusion ``$\supseteq$'', let $(\lambda_1,\ldots,\lambda_k)\in\Lambda(\mR^\all,1)$. By definition, there exist a collection of real numbers 
$$ \{\lambda_{i,R} \mid i \in \{1,\ldots,k\}, \, R \in \mR^\all_i\},$$
satisfying \eqref{c1}--\eqref{c3}. We will show
that $(\lambda_1,\ldots,\lambda_k)\in\Lambda(\mR^{\textnormal{min}},1)$. 
In order to do so, for each $i\in\{1,\ldots,k\}$
and $R\in\mR^\all_i$ we fix a set 
$\mbox{min}_i(R) \in \mR_i^{\textnormal{min}}$
such that $R \supseteq \mbox{min}_i(R)$.
For $Q\in\mR_i^{\textnormal{min}}$ define 
$$\lambda_{i,Q}:=\sum_{\substack{R \in \mR^\all_i \\ \textnormal{min}_i(R)=Q}}\lambda_{i,R}.$$
Then 
the numbers $\{\lambda_{i,Q} \mid i \in \{1,\ldots,k\}, \, Q \in \mR_i^{\textnormal{min}}\}$ satisfy~\eqref{c1}--\eqref{c3}. 
\end{proof}

\section{A Total Capacity Bound}\label{sec:iv}

We here put forward a simple but powerful idea
to obtain outer bounds 
for the service rate region of a coded system.
As an application of this idea, we derive two dual-distance-type bounds in this context.

\begin{lemma}[Total Capacity Bound] \label{tcb}
Let $\mR$ be a recovery $G$-system and let $\{\lambda_{i,R}\}$ be a feasible allocation for $(\mR,1)$.
We have 
\begin{align} \label{new1}
\sum_{i=1}^k \sum_{R\in\mR_i} |R|   \lambda_{i,R} \le n.
\end{align}
\end{lemma}
\begin{proof}
The bound is obtained by summing  the inequalities in \eqref{c2}, for $1 \le j \le n$,
\end{proof}

\begin{remark}
The Total Capacity Bound of Lemma~\ref{tcb} put in evidence the cardinality of the recovery sets.
It can be applied directly when a lower bound for this quantity is known. More precisely, if every recovery set of a $G$-system $\mR$ has size at least~$M$, then
every $(\lambda_1,\ldots,\lambda_k) \in \Lambda(\mR,1)$ satisfies
$$\lambda_1+\cdots + \lambda_k \leq n/M.$$ 
\end{remark}


Lemma~\ref{tcb} motivates us
to determine
the properties of $G$ that determine 
the cardinalities 
of the recovery sets,
and how they shape
the service rate region. This is the program we initiate in this paper.

\subsection{First Dual Distance Bound}
We next establish a connection between the recovery sets of a systematic coded system and the codewords of the code whose parity-check matrix is $G$. By combining this with the Total Capacity Bound, we obtain our First Dual Distance Bound.

\begin{notation}
In the sequel, we denote by 
$\mC$ the linear code generated by $G$, and by $\mC^\perp$, its dual code. Note that $\mC$
is an $[n,k]_q$ code. We let $d^\perp$ be the minimum distance of $\mC^\perp$.
Recall that the (\textbf{Hamming}) \textbf{support} of a vector $x\in\F_q^n$ is $\sigma(x):=\{1 \le i \le n \mid x_i \ne 0\}$.
\end{notation}

\begin{proposition}\label{prop:dualcodewords}
Suppose that $G$ is systematic.
 Let $R \subseteq \{1,\dots,n\}$ and $i \in \{1,\dots,k\}$. Then $R \in \mR^\all_i$ if and only if $R=\{i\}$ or there exists a codeword $x \in \mC^\perp$ with 
 $\sigma(x)\subseteq R \cup \{i\}$
 and $i \in \sigma(x)$.
\end{proposition}
\begin{proof} 
By definition, $R \in \mR^\all_i$ if and only if 
the span of the columns of $G$ indexed by $R\cup \{i\}$
contains $e_i$. This happens if and only if 
$R=\{i\}$ or there exists a linear combination of the columns of~$G$ indexed by~$R$, where the $i$th column of~$G$ is taken with a nonzero coefficient, that gives zero. The latter condition is equivalent to the existence of a codeword $x \in \mC^\perp$ with $i \in \sigma(x)$ and $\sigma(x) \subseteq R$.
\end{proof}


\begin{corollary}\label{cor:size_recovery}
If $G$ is systematic,
$i\in\{1,\ldots, k\}$ and 
$R\in\mR^\all_i$, then $R=\{i\}$ or $|R|\geq d^\perp-1$.
\end{corollary}

We are now ready to present the First Dual Distance Bound. 

\begin{theorem}[First Dual Distance Bound]\label{thm:bounddualdis}
Suppose that~$G$ is systematic and let $\mR$ be a recovery $G$-system.
If $(\lambda_1,\dots,\lambda_k) \in \Lambda(\mR,1)$, then
$$\sum_{i=1}^k \Bigl( \min\{\lambda_i,1\} + (d^\perp-1) \max\{0,\lambda_i-1\}  \Bigr) \le n.$$
\end{theorem}
\begin{proof}
Let $(\lambda_1,\dots,\lambda_k) \in \Lambda(\mR,1)$ and let
$\{\lambda_{i,R}\}$ be a feasible allocation.
By Corollary \ref{cor:size_recovery},
for every $i\in\{1,\ldots, k\}$ 
and every set~$R\in\mR_i$ with $R\ne\{i\}$ we have $|R|\geq d^\perp-1$. We can therefore rewrite and bound the LHS of \eqref{new1} as follows:
\begin{align} \allowdisplaybreaks \label{ee2}
\sum_{i=1}^k &\lambda_{i,\{i\}} + \sum_{i=1}^k \sum_{\substack{R\in\mR_i\\ R \ne \{i\}}} |R| \lambda_{i,R} \nonumber\\ &\geq 
\sum_{i=1}^k \lambda_{i,\{i\}} + ( d^\perp-1) \sum_{i=1}^k \sum_{\substack{R\in\mR_i\\ R \ne \{i\}}}\lambda_{i,R}\nonumber\\ &=
 \sum_{i=1}^k \lambda_{i,\{i\}} + ( d^\perp-1) \sum_{i=1}^k (\lambda_i-\lambda_{i,\{i\}})
\nonumber\\ &=
(d^\perp-1) \sum_{i=1}^k\lambda_i-(
d^\perp-2) \sum_{i=1}^k \lambda_{i,\{i\}}.
\end{align}
Since $G$ has no all-zero column, we have $d^\perp \ge 2$. Therefore, since
$\lambda_{i,\{i\}} \le \min\{\lambda_i,1\}$ for all~$i$, we can further say that the right-hand side of~\eqref{ee2} is at least
\begin{align*}
&(d^\perp-1)\sum_{i=1}^k\lambda_i-(d^\perp-2) \sum_{i=1}^k\min\{\lambda_i,1\} \\
&=\sum_{i=1}^k \Bigl( \min\{\lambda_i,1\} + (d^\perp-1)  \max\{0,\lambda_i-1\}  \Bigr),
\end{align*}
which, combined with \eqref{new1},
gives the statement.
\end{proof}

We can apply Theorem~\ref{thm:bounddualdis} to families of systematic codes whose dual distance is known. For example, by considering the class of systematic MDS codes 
we obtain the following result.

\begin{corollary}
\label{cor:MDSregion}
Suppose that $G$ is systematic and that it generates an MDS code.
Then for all
 $(\lambda_1,\dots,\lambda_k) \in \Lambda(\mR,1)$ we have
$$\sum_{i=1}^k \Bigl( \min\{\lambda_i,1\} + k \cdot \max\{0,\lambda_i-1\}  \Bigr) \le n.$$
\end{corollary}

\begin{remark}
The previous corollary on MDS codes is  sharp whenever $k\leq n-k$; see~\cite[Theorem 2]{service:AktasAJ17}.
\end{remark}


\subsection{Second Dual Distance Bound} \label{sec:fundp}
The goal of this subsection is to identify and study 
new, non-classical parameters of the matrix $G$ that play a role in determining the associated service rate region. As an application, we obtain a second outer bound for the service rate region of a coded system,
which refines the First Dual Distance Bound of Theorem~\ref{thm:bounddualdis} under some assumptions.

\begin{notation} \label{keyparam}
In the sequel, for all
$i \in \{1,\ldots,k\}$, we let
$$G_i=\Bigl( G \mid e_i^\top \Bigr) \in \F_q^{k \times (n+1)}$$
be the matrix obtained from $G$ by appending the $i$th standard basis vector as the $(n+1)$-th column.
Then~$G_i$ generates an $[n+1,k]_q$ code, which we denote by $\mC_i \subseteq \F_q^{n+1}$. Its dual is~$\mC_i^\perp$. For $i \in \{1,\ldots,k\}$, we introduce the following sets and parameters:
\begin{align*}
    \Gamma_i &:= \{\wt(x) \mid x \in \mC_i^\perp, \, n+1 \in \sigma(x)\}, \\ 
    \gamma_i &:= |\Gamma_i| ~~\text{and}~~
    \delta_i^1 :=\min (\Gamma_i). 
\end{align*}
Moreover, we let
\begin{equation*}
    \delta_i^2 := 
    \min \{\wt(x) \mid x \in \mC_i^\perp, \, n+1 \in \sigma(x), \wt(x)>\delta_i^1\},
\end{equation*}
with the convention that $\delta_i^2=\delta_i^1$ when $\gamma_i=1$.
Finally, we define
$$\omega_i:=\frac{|\{x\in\mC_i^\perp : n+1\in\sigma(x), \, \wt(x)=\delta_i^1\}|}{q-1}.$$
\end{notation}

The structural parameters defined in Notation~\ref{keyparam} play an essential role in determining the corresponding service rate region.
In this conference paper, we give preliminary evidence of this with a result that 
extends the First Dual Distance Bound of Theorem~\ref{thm:bounddualdis} to possibly non-systematic matrices. 

We start by describing the recovery sets of $\mR^\all$ in terms of the codes $\mC_i$, $i \in \{1,\ldots,k\}$. The proof of the next result is similar to the one of Proposition~\ref{prop:dualcodewords}, and we, therefore, omit it.

\begin{proposition}\label{prop:recoverysize}
Let $R\subseteq\{1,\ldots,n\}$ and $i\in\{1,\ldots,k\}$. Then $R\in\mR^\all_i$ if and only if there is~$x\in\mC_i^\perp$ with
$\sigma(x)\subseteq R\cup\{n+1\}$
and 
$n+1\in\sigma(x)$. In particular, if $R\in\mR^\all_i$ then $|R|\geq \delta_i^1-1$. 
\end{proposition}

Proposition \ref{prop:recoverysize} means that there are~$\omega_1$ sets $R\in\mR^\all_i$ with size $\delta_i^1$. All the other sets have size at least $\delta_i^2$. 

\begin{theorem}[Second Dual Distance Bound]\label{thm:general}
Let~$\mR$ be a recovery $G$-system and let $(\lambda_1,\ldots, \lambda_k)\in\Lambda(\mR,1)$. For each $i \in \{1,\ldots,k\}$, define $\ell_i:=\min\{\lambda_i,1\}$. Then
$$\sum_{i=1}^k \Bigl((\delta_i^2-1)(\lambda_i-\omega_i\ell_i) + (\delta_i^1-1)\, \omega_i \ell_i\Bigr)\leq n.$$

\end{theorem}
\begin{proof} 
Let $(\lambda_1,\ldots, \lambda_k)\in\Lambda(\mR,1)$ and let
 $\{\lambda_{i,R}\}$ be a feasible allocation. By Proposition~\ref{prop:recoverysize} and~\eqref{c1} we have that, for all $i \in \{1,\ldots,k\}$,
\begin{equation} \label{prel}
\sum_{\substack{R\in\mR_i\\|R|>\delta_i^1-1}}\lambda_{i,R}=\lambda_i- \sum_{\substack{R\in\mR_i\\|R|=\delta_i^1-1}}\lambda_{i,R},
\end{equation}
where the sum over an empty set is $0$.
Thus the LHS of~\eqref{new1} can be rewritten as 
\begin{align} \label{new2}
\sum_{i=1}^k\Bigl( (\delta_i^1-1)\sum_{\substack{R\in\mR_i \\|R|=\delta_i^1-1}}\hspace{-7pt}\lambda_{i,R}+ \sum_{\substack{R\in\mR_i \\|R|>\delta_i^1-1}} \hspace{-7pt} |R| \, \lambda_{i,R}  \Bigr).
\end{align}
Using~\eqref{prel}, Proposition~\ref{prop:recoverysize}, and the definition of $\delta_i^2$, we find that
the quantity in~\eqref{new2} is at least
\begin{align}
\allowdisplaybreaks
       &\sum_{i=1}^k\Bigl ((\delta_i^1-1)\sum_{\substack{R\in\mR_i \\|R|=\delta_i^1-1}}\hspace{-7pt}\lambda_{i,R}+ (\delta_i^2-1)\sum_{\substack{R\in\mR_i \\|R|>\delta_i^1-1}} \hspace{-7pt} \lambda_{i,R}  \Bigr)    \nonumber\\
     =&\sum_{i=1}^k\Bigl((\delta_i^2-1)\lambda_i - (\delta_i^2-\delta_i^1)\sum_{\substack{R\in\mR_i\\|R|=\delta_i^1-1}}\hspace{-7pt}\lambda_{i,R}\Bigr)  \nonumber\\
     \geq&\sum_{i=1}^k\Bigl((\delta_i^2-1)\lambda_i - (\delta_i^2-\delta_i^1)\omega_i\ell_i\Bigr),\label{new3}
\end{align}
where the latter inequality follows from the fact that $\delta_i^2-\delta_i^1 \ge 0$ and $\lambda_{i,R} \le \ell_i$
for all $i \in \{1,\ldots,k\}$.
Finally, the last line of~\eqref{new3}
is equal to $$\sum_{i=1}^k \Bigl((\delta_i^2-1)\Bigl(\lambda_i-\omega_i\ell_i\Bigr) + (\delta_i^1-1)\omega_i\ell_i\Bigr),$$
which with~\eqref{new1}~and~\eqref{new2} concludes the proof.
\end{proof}

We illustrate the previous result with two examples where $G$ is not systematic.

\begin{example}[First-Order Reed-Muller Code]\label{ex:reed_muller}
Consider the following non-systematic generator matrix $G\in\F_2^{4\times 8}$ of the first order Reed-Muller code, namely
$$ G:=\begin{pmatrix}
1 & 1 & 1 & 1 & 0 & 0 & 0 & 0 \\
1 & 0 & 1 & 0 & 1 & 0 & 1 & 0 \\
1 & 1 & 0 & 0 & 1 & 1 & 0 & 0 \\
1 & 1 & 1 & 1 & 1 & 1 & 1 & 1 
\end{pmatrix}.$$

Consider the codes generated by $G_i$, for $i\in\{1,2,3,4\}$. We have that $\delta_1^1 = \delta_2^1=\delta_3^1=3$, $\omega_1=\omega_2=\omega_3=4$, $\delta_1^2=\delta_2^2=\delta_2^2=5$,
$\delta_4^1=2$, $\omega_4=1$ and $\delta_4^2=4$. Thus Thm~\ref{thm:general} reads,
with $\ell_i=\min\{\lambda_i,1\}$ for all~$i$,
\begin{align}\label{eq:reedMuller_our_bound}
    4\lambda_1+4\lambda_2+4\lambda_3 -8\ell_1 - 8\ell_2-8\ell_3+3\lambda_4-2\ell_4\leq 8.
\end{align}
The bound is not sharp in this case. Indeed, suppose that $\lambda_2=\lambda_3=0$. 
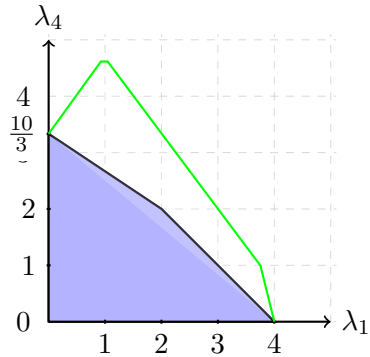
\begin{figure}[hbt]
\begin{center}
\begin{tikzpicture}[ thick,scale=0.75]
\draw[help lines, color=gray!30, dashed] (-.01,-.01) grid (5.1,5.1);
\draw[->, thick,black] (0,0)--(5,0) node[right]{$\lambda_1$};
\draw[->, thick,black] (0,0)--(0,5) node[above]{$\lambda_4$};
\draw[thick,blue] (0,3.3333)--(2,2)--(4,0);
\node[circle,scale=.5,label=left:$0$] (O) at (0,0) {};
\node[circle,scale=.5,label=below:$1$] (x1) at (1,0) {};
\node[circle,scale=.5,label=below:$2$] (x2) at (2,0) {};
\node[circle,scale=.5,label=below:$3$] (x3) at (3,0) {};
\node[circle,scale=.5,label=below:$4$] (x4) at (4,0) {};
\node[circle,scale=.5,label=left:$1$] (y1) at (0,1) {};
\node[circle,scale=.5,label=left:$2$] (y2) at (0,2) {};
\node[circle,scale=.5,label=left:$3$] (y3) at (0,3) {};
\node[circle,scale=.5,label=left:$4$] (y4) at (0,4) {};
\draw[fill=blue!30]  (0,0) -- (4,0) -- (2,2) -- (0,3.3333)  -- cycle;
\foreach \x/\xtext in {1, 2, 3, 4}
   \draw (\x cm,1pt) -- (\x cm,-1pt) node[anchor=north,fill=white] {$\xtext$};
 \foreach \y/\ytext in {1,2, 3.33/\frac{10}{3}}
   \draw (1pt,\y cm) -- (-1pt,\y cm) node[anchor=east,fill=white] {$\ytext$};
    \filldraw[thick,draw=green,fill=white,fill opacity=0.2] (0,10/3) -- (0.9271,4.616) -- (1.046,4.616) -- (3.754,0.9955) -- (4,0);
\end{tikzpicture}
\end{center}
\caption{Service rate region for the $G$-system in Example~\ref{ex:reed_muller}, with $\lambda_2=\lambda_3=0$. Eq.\ \eqref{eq:reedMuller_our_bound} is represented by the green line.}
\end{figure}

\end{example}

\begin{example}
Consider the finite field $\F_8=\F_2[\alpha]$, where 
$\alpha^3+\alpha+1=0$.  Let 
$$ G:=\begin{pmatrix}
1 & \alpha & \alpha^4 &\alpha & \alpha^3 \\
\alpha & \alpha^3 & \alpha &\alpha^6 & \alpha^5 \\
\alpha^2 & \alpha & \alpha^4 & \alpha^2 & \alpha^3
\end{pmatrix}\in\F_8^{3\times 5}.$$
We have $d^\perp=2$. Moreover, the dual distances of the codes generated by $G_i$, for $i\in\{1,2,3\}$, are all equal to $d^\perp$.
Further, we have $\delta_1^1=3$, $\omega_1=1$, $\delta_1^2=4$,
$\delta_2^1=3$, $\omega_2=2$, $\delta_2^2=4$, $\delta_3^1=4$, $\omega_3=7$, and $\delta_3^2=5$. Thus Thm~\ref{thm:general} reads,
with $\ell_i=\min\{\lambda_i,1\}$ for all~$i$,
\begin{align*}
    3\lambda_1+3\lambda_3+4\lambda_3 -\ell_1 - 2\ell_2-7\ell_3\leq 5.
\end{align*}
\end{example}

We conclude this section by showing that Theorem~\ref{thm:general} implies
the First Dual Distance Bound of Theorem~\ref{thm:bounddualdis}
when $G$ is systematic and $d^\perp \ge 3$.

\begin{proposition}
Suppose that $G$ is systematic and that
$d^\perp \ge 3$. Then for all $i \in \{1,\ldots,k\}$ we have $\delta_i^1=2$,  $\omega_i=1$, and $\delta_i^2 \ge d^\perp$.
\end{proposition}

\begin{proof}
We only prove the result for $i=1$; the proof for the other indices is identical.
We denote by $\pi:\F_q^{n+1} \to \F_q^n$ the projection onto the first $n$ coordinates.
That $\delta_1^1=2$ easily follows from the definitions, since
the $1$st and the $(n+1)$th columns of $G_1$ are equal.

We claim that $(1,0,\ldots,0,-1) \in \F_q^{n+1}$ is the only codeword 
$x \in \mC_1^\perp$ (up to multiples) with $n+1 \in \sigma(x)$ and Hamming weight 2.
Let $x=(1,0,\ldots,0,-1)$ and 
suppose towards a contradiction that there exists $y \in \mC_i^\perp$ with
$y_{n+1}=1$, Hamming weight 2, and linearly independent from $x$.
Then $\pi(x+y) \in \mC^\perp$ is non-zero and has weight at most~2, contradicting $d^\perp \ge 3$.
This shows that $\omega_1=1$.

It remains to show that $\delta_1^2 \ge d^\perp$. Let $x \in \mC_1^\perp$ be the codeword defined above.

\begin{itemize}[leftmargin=*]
    \item We start by showing that $\delta_1^2 \ge 3$.
Take $y \in \mC^\perp$
with $\wt(y)=d^\perp \ge 3$.
Then $(y,0)+x \in \mC_i^\perp$ has weight at least $d^\perp \ge 3$.
Not all codewords of $\mC_1^\perp$ with $n+1$ in their support have weight 2. Therefore $\delta_1^2 \ge 3$ by definition of $\delta_1^2$.

\item Let $y\in\mC_1^\perp$ be a codeword with $\wt(y)=\delta_1^2$ and $n+1\in\sigma(y)$. By the previous item, $\wt(y) \ge 3$.

If $1\not\in\sigma(y)$, then  $\pi(x+y) \in \mC$ is a non-zero codeword with the same weight as $y$. Therefore $\delta_i^2=\wt(y) \ge d^\perp$. On the other hand, if $1\in\sigma(y)$ then  $z:=(y_1+y_{n+1},y_2,\ldots,y_n) \in \mC^\perp$. Since $\wt(y) \ge 3$ we have $z \neq 0$. Therefore
$d^\perp \le \wt(z) \le \wt(y) = \delta_1^2$. \qedhere
\end{itemize}
\end{proof}

Observe that when $G$ is systematic and $d^\perp=2$,
Theorem~\ref{thm:general} can be sharper than Theorem~\ref{thm:bounddualdis},
since it takes into account possibly finer information.

\begin{example}
Let $$ G:=\begin{pmatrix}
1 & 0 & 0 & 3 & 5 \\
0 & 1 & 0 & 0 & 1 \\
0 & 0 & 1 & 0 & 3
\end{pmatrix}\in \F_7^{3\times 5}.$$ 
The bound of Theorem \ref{thm:bounddualdis} reads
$\lambda_1+\lambda_2+\lambda_3\leq 5$,
while the one of
Theorem~\ref{thm:general} reads
$2\lambda_1+3\lambda_2 +3\lambda_3 -2\sum_{i=1}^3\min\{\lambda_i,1\}\leq 5.$
It can be shown that the latter region is strictly 
contained in the former.
\end{example}

The example illustrates that the service rate region of a $G$-system depends on $G$. How invertible operations on $G$ change the region is an open problem.

\section{Conclusions and Future Work}
\label{sec:v}
The service rate region is a new aspect of linear redundancy schemes, which measures their ability to provide simultaneous data recovery. The problem was studied in different frameworks. We initiates a study within the classical coding theory. This approach enabled us to derive some new results and recover some previously known bounds on the service rate region straightforwardly.
A liner redundancy system is defined by a matrix akin to a code generator matrix.
We believe that coding theory can be instrumental in identifying and using the matrix's structural properties to construct polytopes that contain the service rate region by establishing a series of outer bounds that, when combined, determine the region.


\bigskip

\bibliographystyle{IEEEtran}
\bibliography{serviceL}

\begin{thebibliography}{10}
\providecommand{\url}[1]{#1}
\csname url@samestyle\endcsname
\providecommand{\newblock}{\relax}
\providecommand{\bibinfo}[2]{#2}
\providecommand{\BIBentrySTDinterwordspacing}{\spaceskip=0pt\relax}
\providecommand{\BIBentryALTinterwordstretchfactor}{4}
\providecommand{\BIBentryALTinterwordspacing}{\spaceskip=\fontdimen2\font plus
\BIBentryALTinterwordstretchfactor\fontdimen3\font minus
  \fontdimen4\font\relax}
\providecommand{\BIBforeignlanguage}[2]{{%
\expandafter\ifx\csname l@#1\endcsname\relax
\typeout{** WARNING: IEEEtran.bst: No hyphenation pattern has been}%
\typeout{** loaded for the language `#1'. Using the pattern for}%
\typeout{** the default language instead.}%
\else
\language=\csname l@#1\endcsname
\fi
#2}}
\providecommand{\BIBdecl}{\relax}
\BIBdecl

\bibitem{Edge:YadgarKAS19}
G.~Yadgar, O.~Kolosov, M.~F. Aktas, and E.~Soljanin, ``Modeling the edge:
  Peer-to-peer reincarnated,'' in \emph{2nd {USENIX} Workshop on Hot Topics in
  Edge Computing, HotEdge 2019, Renton, WA, USA, July 9, 2019}, I.~Ahmad and
  S.~Sundararaman, Eds.\hskip 1em plus 0.5em minus 0.4em\relax {USENIX}
  Association, 2019.

\bibitem{service:aktas2021jkks}
M.~Aktas, G.~Joshi, S.~Kadhe, F.~Kazemi, and E.~Soljanin, ``Service rate
  region: A new aspect of coded distributed system design,'' \emph{IEEE Trans.\
  Inform. The}, Feb. 2022.

\bibitem{service:KazemiKSS21g}
F.~Kazemi, S.~Kurz, and E.~Soljanin, ``Efficient storage schemes for desired
  service rate regions,'' in \emph{2021 IEEE Information Theory Workshop
  (ITW)}, Apr. 2021.

\bibitem{service:KazemiKSS20g}
------, ``A geometric view of the service rates of codes problem and its
  application to the service rate of the first order reed-muller code,'' in
  \emph{2020 IEEE International Symposium on Information Theory (ISIT)}, June
  2020.

\bibitem{service:KazemiKSS20c}
F.~Kazemi, E.~Karimi, E.~Soljanin, and A.~Sprintson, ``A combinatorial view of
  the service rates of codes problem, its equivalence to fractional matching
  and its connection with batch codes,'' in \emph{2020 IEEE International
  Symposium on Information Theory (ISIT)}, June 2020.

\bibitem{service:AndersonJJ18}
S.~E. Anderson, A.~Johnston, G.~Joshi, G.~L. Matthews, C.~Mayer, and
  E.~Soljanin, ``Service capacity region of content access from erasure coded
  storage,'' in \emph{IEEE Information Theory Workshop (ITW)}, Nov. 2018.

\bibitem{service:AktasAJ17}
M.~Akta{\c{s}}, S.~E. Anderson, A.~Johnston, G.~Joshi, S.~Kadhe, G.~L.
  Matthews, C.~Mayer, and E.~Soljanin, ``On the service capacity region of
  accessing erasure coded content,'' in \emph{2017 55th Annual Allerton
  Conference on Communication, Control, and Computing (Allerton)}, 2017, pp.
  17--24.

\bibitem{BatchCodesAndTheirApps:IshaiKO04}
Y.~Ishai, E.~Kushilevitz, R.~Ostrovsky, and A.~Sahai, ``Batch codes and their
  applications,'' in \emph{Proceedings of the 36th Annual {ACM} Symposium on
  Theory of Computing, Chicago, IL, USA, June 13-16, 2004}, L.~Babai, Ed.,
  2004, pp. 262--271.

\bibitem{Batch:RietST18}
A.~Riet, V.~Skachek, and E.~K. Thomas, ``Asynchronous batch and {PIR} codes
  from hypergraphs,'' in \emph{{IEEE} Inform.\ Theory Workshop, {ITW} 2018,
  Guangzhou, China, November 25-29, 2018}.\hskip 1em plus 0.5em minus
  0.4em\relax {IEEE}, 2018, pp. 1--5.

\bibitem{batchPIR:Skachek18}
V.~Skachek, \emph{Batch and PIR Codes and Their Connections to Locally
  Repairable Codes}.\hskip 1em plus 0.5em minus 0.4em\relax Springer
  International Publishing, 2018.

\bibitem{PIR:Fazeli15-1}
\BIBentryALTinterwordspacing
A.~Fazeli, A.~Vardy, and E.~Yaakobi, ``{PIR} with low storage overhead: Coding
  instead of replication,'' \emph{CoRR}, vol. abs/1505.06241, 2015. [Online].
  Available: \url{http://arxiv.org/abs/1505.06241}
\BIBentrySTDinterwordspacing

\bibitem{PIR:Fazeli15-2}
A.~{Fazeli}, A.~{Vardy}, and E.~{Yaakobi}, ``Codes for distributed pir with low
  storage overhead,'' in \emph{2015 IEEE Internat.\ Symp.\ on Inform.\ Theory
  (ISIT)}, 2015, pp. 2852--2856.

\bibitem{Ernvall:14}
\BIBentryALTinterwordspacing
T.~Ernvall, T.~Westerb{\"{a}}ck, and C.~Hollanti, ``Linear locally repairable
  codes with random matrices,'' \emph{CoRR}, vol. abs/1408.0180, 2014.
  [Online]. Available: \url{http://arxiv.org/abs/1408.0180}
\BIBentrySTDinterwordspacing

\bibitem{Ernvall:14a}
T.~Ernvall, T.~Westerback, and C.~Hollanti, ``Constructions of optimal and
  almost optimal locally repairable codes,'' in \emph{Wireless Communications,
  Vehicular Technology, Information Theory and Aerospace Electronic Systems
  (VITAE), 2014 4th International Conference on}, May 2014, pp. 1--5.

\bibitem{Gopalan:14}
P.~{Gopalan}, C.~{Huang}, B.~{Jenkins}, and S.~{Yekhanin}, ``Explicit maximally
  recoverable codes with locality,'' \emph{IEEE Trans.\ Inform. Th.}, vol.~60,
  no.~9, pp. 5245--5256, 2014.

\bibitem{OptimalLRC}
I.~Tamo and A.~Barg, ``A family of optimal locally recoverable codes,''
  \emph{IEEE Transactions on Information Theory}, vol.~60, no.~8, pp.
  4661--4676, Aug 2014.

\end{thebibliography}

\end{document}